\documentclass[10pt,oneside,a4paper,reqno]{amsart}  
\usepackage{bbm}
\usepackage{geometry}                
\usepackage{latexsym} 
\usepackage{amsfonts} 
\usepackage{upgreek}
\usepackage{amsmath} 
\usepackage{amsthm} 
\usepackage{mathrsfs} 
\usepackage{dsfont} 
\usepackage{bbold} 
\usepackage[english]{babel} 
\usepackage{caption} 
\usepackage{epsfig} 
\usepackage{float} 
\usepackage{subfigure} 
\usepackage{psfrag} 
\usepackage{graphicx,xcolor} 
\usepackage{epsfig} 
\usepackage{amsbsy} 
\usepackage{amssymb}
\usepackage{color} 
\usepackage[colorlinks=true, linkcolor=magenta, citecolor=cyan, urlcolor=blue]{hyperref}

\DeclareMathOperator{\Val}{\matV}

 \DeclareMathOperator{\sign}{sign}

\theoremstyle{plain}
\newtheorem{theorem}{Theorem}[section]
\newtheorem{lemma}[theorem]{Lemma}
\newtheorem{prop}[theorem]{Proposition}

\theoremstyle{remark}
\newtheorem{rmk}[theorem]{Remark}

 \definecolor{light}{gray}{.9}

\newcommand{\beq}{\begin{equation}}
\newcommand{\eeq}{\end{equation}}
\def\bea{\begin{eqnarray}}
\def\eea{\end{eqnarray}}

\renewcommand{\Re}{{\rm Re}}
\renewcommand{\Im}{{\rm Im}}
 
\newcommand{\ZZZ}{\mathds{Z}}

\newcommand{\RRR}{\mathds{R}} 
\newcommand{\TTT}{\mathds{T}} 
\newcommand{\uno}{\mathds{1}} 
 
\newcommand{\DD}{{\mathcal D}} 
 
\newcommand{\calF}{{\mathcal F}} 
\newcommand{\calG}{{\mathcal G}}

\newcommand{\MM}{{\mathcal M}} 
 
\newcommand{\calO}{{\mathcal O}}

\newcommand{\RR}{{\mathcal R}} 
 
\newcommand{\TT}{{\mathcal T}}

\newcommand{\ol}{\overline} 
 
\newcommand{\io}{\infty} 
 
\newcommand{\al}{\alpha} 
 
\newcommand{\be}{\beta} 
 
\newcommand{\m}{\mu} 
\newcommand{\x}{\xi}

\newcommand{\g}{\gamma} 
\newcommand{\om}{\omega}

\newcommand{\f}{\varphi} 
\newcommand{\s}{\sigma} 
 
\newcommand{\del}{\partial}

\newcommand{\nn}{\boldsymbol{n}}

\newcommand{\ee}{\boldsymbol{e}}

\newcommand{\ii}{{\rm i}}

\def\tilde#1{\widetilde{#1}}
 
\def\ins#1#2#3{\vbox to0pt{\kern-#2 \hbox{\kern#1 #3}\vss}\nointerlineskip}

\numberwithin{equation}{section}

\def\be{\begin{equation}}
\def\ee{\end{equation}}
\def\bea{\begin{eqnarray}}
\def\eea{\end{eqnarray}}
\def\nn{\nonumber}

\def\sx{\sigma^x}
\def\sy{\sigma^y}
\def\sz{\sigma^z}

\def\a{\alpha}
\def\d{\delta}
\def\g{\gamma}

\def\D{\Delta}

\def\r{\rho}

\def\t{\tau}

\def\w{\omega}

\def\sign{\operatorname{sign}}
\def\Z{\mathds{Z}}

\def\f{\varphi}
\def\hp{\psi}

\def\R{\mathds{R}}
\def\N{\mathds{N}}

\newcommand{\OOO}[1]{O \left(#1\right)}
\newcommand{\OO}[1]{O \left(\frac{1}{#1}\right)}

\newcommand{\nocontentsline}[3]{}
\newcommand{\tocless}[2]{\bgroup\let\addcontentsline=\nocontentsline#1{#2}\egroup}

\DeclareMathSymbol{\leqslant}{\mathalpha}{AMSa}{"36} 
\DeclareMathSymbol{\geqslant}{\mathalpha}{AMSa}{"3E} 
\DeclareMathSymbol{\eset}{\mathalpha}{AMSb}{"3F}     
\renewcommand{\leq}{\;\leqslant\;}                   
\renewcommand{\geq}{\;\geqslant\;}                   
\renewcommand{\le}{\;\leqslant\;}                   
\renewcommand{\ge}{\;\geqslant\;}                   
\newcommand{\dd}{\,\text{\rm d}}

\def\Val{\operatorname{Val}}

\begin{document}
 
\title[Long time behaviour]{\bf 
Long time behaviour of a local perturbation in the isotropic XY chain under periodic forcing}

\author{Livia Corsi}
\address{Dipartimento di Matematica e Fisica, Universit\`a di Roma Tre, Roma, I-00146, Italy}
\email{lcorsi@mat.uniroma3.it}

\author{Giuseppe Genovese}
\address{Department Mathematik und Informatik, Universit\"at Basel Spiegelgasse 1, CH-4051 Basel, Switzerland}
\email{giuseppe.genovese@unibas.ch}


\date{\today} 
 
\maketitle 
\begin{abstract}
We study the isotropic XY quantum spin chain with a time-periodic transverse magnetic field acting on a single site. The asymptotic problem can be mapped into a highly resonant Floquet-Schr\"odinger equation, for which, under a diophantine-like assumption on the frequency, we show the existence of a periodic solution. The proof is based on a KAM-type renormalisation. This in turn implies the state of the quantum spin chain to be asymptotically a periodic function synchronised with the forcing also at arbitrarily low frequencies. 

\medskip

\noindent
\textbf{MSC:} 82C10, 37K55, 45D05.
\end{abstract}


  

\section{Introduction}\label{intro} 

We investigate the isotropic XY
 quantum spin chain with a periodically time-dependent transverse external field acting only on 
 one site, namely the $\kappa$-th, 
 and free boundary conditions. The Hamiltonian reads
\be\label{eq:HXX-Vt}
H_N(t)=-g\sum_{j=1}^{N-1}\left(\sx_j\sx_{j+1}+\sy_{j}\sy_{j+1}\right)-hV(\w t)\sz_\kappa\,,\quad 1<\kappa<N\,.
\ee

Here $\sx,\sy,\sz$ denote the Pauli matrices, $g,h,\w>0$ are parameters ruling respectively 
the spin-spin coupling, 
the magnitude of the external field and its frequency. We assume that $V(\om t)$ is a real 
periodic analytic function with frequency $\om$:
\be\label{eq:V}
V(\om t)=\sum_{\substack{k\in\ZZZ}} e^{\ii k \om t} V_k\,,
\qquad |V_k|\le C_0 e^{-\s |k|}\,.
\ee

For any $t\in\R$ and $N\in\N$ $H_N(t)$ is a self-adjoint operator on 
$\mathcal{H}_N:={\mathbb{C}^2}^{\otimes N}$, and the 
thermodynamic limit $N\to\infty$ is done as customary in the Fock space
 $\mathcal{F}:=\bigoplus_N \mathcal{H}_N$. 

It is well-known that this system is equivalent to a chain of quasi-free fermions 
and therefore the $N$-particle state is fully described 
by a one-particle wave function. At fixed $t$ the forcing $V(\w t)$ is just a number which we can
incorporate into $h$ and the spectrum is given by the standard analysis of the rank-one perturbation of the 
Laplacian on $\Z$ (see \cite{ABGM1}). Precisely, as $N\to\infty$ we have a band $[-g,g]$ and 
a isolated eigenvalue given by
\be\label{eq:spectHn}
g\sign(h)\sqrt{1+\frac{h^2}{g^2}}\,.
\ee

The study of the dynamics however is not as simple, because when $t$ varies the eigenvalue
moves, and it can touch the band creating resonances. 
More precisely, the dynamics in the time interval $[t_0,t]$ is governed by the following 
Floquet-Schr\"odinger equation on $\Z$ 
(for details on its derivation and its relation with the many-body system we refer to \cite{ABGM1, ABGM2, GG, CG})
\begin{equation}\label{THE EQ}
\ii\del_t\psi(x,t)=gh\Delta\psi(x,t)+hH_F(t,t_0)\psi(x,t)\,,\quad \psi(x,t_0)=\d(x)\,,\quad x\in\Z\,.
\end{equation}
Here $\d(x)$ is the Kronecker delta centred in the origin, $\Delta$ is the Laplacian on $\Z$ 
with spectrum given by $\{-\cos q, q\in[-\pi,\pi]\}$, and
\be
H_F(t,t_0)\psi(x,t):= V(\om t)\psi(x,t)+ig\int_{t_0}^{t}\dd t' J_1(g(t-t'))e^{-\ii\Delta (t-t')}V(\w t')\psi(x, t')\,,
\ee
where 
$$
J_{k}(t):=\frac{1}{2\pi}\int_{-\pi}^{\pi}dx e^{\ii xk+\ii t\cos x }\,,\quad k\in\Z\,
$$

The Floquet operator $H_F$ acts as a memory-term, accounting for the retarded effect of the rest 
of the chain on the site $\kappa$. 
This equation finds a more compact form in the Duhamel representation in the momentum space. 
We denote by $\xi\in[-1,1]$ the points 
of the spectrum of $-\D$. Moreover with a slight abuse of notation throughout the paper we will 
systematically omit the customary $\hat \cdot$ to 
indicate either Fourier transforms (when transforming in space) and Fourier coefficients 
(when transforming in time). 
 
Let $\psi(\xi,t)$, $\xi\in[-1,1]$, denote the Fourier transform of $\psi(x,t)$,  $x\in\Z$. The corresponding of 
equation (\ref{THE EQ}) for $\psi(\xi,t)$ in its Duhamel form reads
\be\label{duat0}
(\uno+\ii hW_{t_0})\psi(\x,t)=1\,,
\ee
where $\{W_{t_0}\}_{t_0\in\R}$ is a family of Volterra operators for any $t>t_0$ and $\xi\in[-1,1]$, defined via
\be\label{eq:Wt0}
W_{t_0}f(\x,t):=\int_{t_0}^t\dd t' J_0(g(t-t'))e^{\ii g\x(t-t')}V(\w t')f(\x, t')\,.
\ee
Let $L^{2}_\xi C^\w_t([-1,1]\times[t_0,t])$ denote the space of square integrable functions in $[-1,1]$
and real analytic\footnote{Mind here the superscript $\w$ denoting analyticity as
customary, not to be confused with the frequency.} in the time interval $[t_0,t]$ .
 Each $W_{t_0}$ is a linear map from 
$L^{2}_\xi C^\w_t([-1,1]\times[t_0,t])$ into itself. For any $t_0$ finite $W_{t_0}$ is a compact integral 
operator, which ensures the existence of a unique solution for $t-t_0<\infty$ (see for instance \cite{EN}). 
We denote this one-parameter 
family of functions with $\psi_{t_0}(\xi,t)$. As $t_0\to-\infty$ the limit of the $W_{t_0}$ is an unbounded 
operator, denoted by $W_{\infty}$, defined through
\be\label{eq:Winf}
W_{\infty}f(\x,t):=\int_{-\infty}^t\dd t' J_0(g(t-t'))e^{\ii g\x(t-t')}V(\w t')f(\x, t')\,.
\ee
One can therefore use $W_{\infty}$ to define an asymptotic version of equation (\ref{duat0}) as $t_0\to-\infty$
\be\label{duat}
(\uno+\ii hW_{\infty})\psi(\x,t)=1\,,
\ee
whose solutions are denoted by $\psi_\infty(\xi,\w t)$: indeed it is easy to check that $W_{\infty}$ maps 
periodic functions of frequency $\w$ 
into periodic functions of frequency $\w$, thus it is somehow expected to find solutions of 
(\ref{duat}) in this class of functions. 
Our main result partially confirms this idea. Precisely
we need the following  assumption on the frequency:
\be\label{eq:dioph-alpha}
\bar{\upepsilon}:=\inf_{k\in\N}\left|\frac{2g}{\om}- k\right|>0\,. 
\ee

Then we have the following result.

\begin{theorem}\label{main-as}
Let $\w>0$ satisfying (\ref{eq:dioph-alpha}). There is $\g_0=\g_0(\om,g,V)$ 
small enough such that if $h<\g_0\om$, 
then there exists a periodic solution of (\ref{duat}) with frequency $\om$, 
$\psi_\infty(x,\w t)\in L_x^2C_t^{\w}(\Z\times\R)$. In particular $\g_0$ is explicitly computable; see \eqref{gamma0}. Moreover
\be\nonumber
\psi_{t_0}(x,t)=\psi_{\infty}(x,\w t)+\OO{\sqrt{t-t_0}}\,.
\ee
\end{theorem}

The main relevance of this result lies in its validity for low frequencies. To the best of our knowledge, a similar control of the convergence to the synchronised periodic state for a periodically forced small quantum system coupled with free fermion reservoirs has been achieved only in \cite{fr}, for a different class of models. In general it is known that the low-frequency assumption makes the dynamics harder to study. 

In \cite{CG} we proved the existence of periodic solutions of (\ref{duat}) with frequency $\w$ if $V_0=\OO{h}$ and 
$h$ small or if $V_0=0$, $\w>2g$ (high frequencies) and $h/\w$ small. The meaning of both these conditions is clear: if $V_0$ is large 
and $h$ is small, then the eigenvalue does not touch the band; if $\w>2g$ then the forcing cannot move energy levels within the band. 
In particular the high-frequency assumption appears in other related works in mathematical and theoretical physics \cite{woj1, BDP, NJP, woj2, woj3, alberto}.

In \cite[Proposition 3.1]{CG} we also proved that if a periodic solutions of (\ref{duat}) 
$\psi_\infty(\xi,t)$ with frequency $\w$ exists, 
then $\psi_{t_0}$ must approach $\psi_\infty(\xi,t)$ as $t_0\to-\infty$, namely the following result.
\begin{prop}\label{prop:asintotica}
Let $\psi_\infty(\x,\w t)$ a periodic solution of (\ref{duat}) with frequency $\w$ satisfying (\ref{eq:dioph-alpha}). 
For any $t\in\RRR$, $\xi\in[-1,1]$ one has
\be\label{eq:delta}
\psi_{t_0}(\x,t)=\psi_\infty(\x,\w t)+\OO{\sqrt{t-t_0}}\,. 
\ee
\end{prop}
Therefore the control on the long time behaviour of the solution of (\ref{duat0}) amounts to establish the existence of a periodic solution of (\ref{duat}) for $\w<2g$, a condition defining the low frequency regime. This is a genuine PDE question, which is indeed the main focus of this paper.

More specifically, we are facing an unbounded time-dependent perturbation of the continuous spectrum of the Laplacian 
on $\Z$. Problems involving periodic forcing are typically dealt with via a KAM-approach, namely one tries to reduce the perturbation to a 
constant operator by means of a sequence of bounded maps. This is for instance the approach adopted in \cite{NJP, woj3} in the context of interacting many-body system, in which a generalisation of the classical Magnus expansion is exploited via normal form methods. Indeed some 
salient features of periodically driven systems, as for instance pre-thermalisation or slow heating,  from the mathematical point of view are 
essentially consequences of the KAM reduction. A similar approach has been used in \cite{alberto} for the Klein-Gordon equation with a quasi-periodic forcing. All the aforementioned results are valid if the frequency is large enough, as usual in Magnus expansion approaches.  

We cope here with two main sources of difficulty.
First we deal with a perturbation of operators with continuous spectrum. Secondly the operator in (\ref{duat}) is a perturbation 
of the identity, which makes trivial the homological equation at each KAM step. Thus we have to use a different approach. 
As in \cite{CG}, we explicitly construct a solution of (\ref{duat}) by resumming the Neumann series. The main difficulty is represented by the 
occurrence of small denominators which actually vanish at some point within the spectrum of the Laplacian and also accumulate in the 
coefficients of the series. We cure these divergences by a suitable renormalisation of the Neumann series and one major advance of this work is that this is done regardless of the size of the frequency $\w$.

The interaction of a small system (the impurity) with an environment (the rest of the chain) while it is irradiated by monochromatic light 
is a question of primary interest in non-equilibrium statistical physics. Although more complicated systems have been considered 
\cite{matt, bru, fr}, quantum spin chains are particularly appealing as they present a rich phenomenology along with a limited amount 
of technical difficulties. Indeed the lack of ergodicity of such systems has already been object of study in the '70s \cite{Leb, rob}. 
The choice of considering the isotropic XY chain itself simplifies greatly the computations, as one gets exact formulas for the functions $J_k$. 
The dynamics of an impurity was first analysed in \cite{ABGM2} with different forms of time-dependent external fields. In particular in 
the case $V(\w t)=\cos\w t$ the authors computed the magnetisation of the perturbed spin at the first order in $h$, observing a divergence at $\w= 2g$, i.e. for a value violating \eqref{eq:dioph-alpha}.

\

The rest of the paper is organised as follows. In Section \ref{eq} the main objects needed for the proof are introduced while in Section \ref{renorm} we 
prove the existence of a periodic solution of (\ref{duat}) with frequency $\w$. In Section \ref{importa} and Section \ref{stimazze} we prove few accessory results used in Section \ref{renorm}. Finally we attach an Appendix in which we sketch the proof of Proposition \ref{prop:asintotica}.

\subsection*{Acknowledgements} The authors thank for the nice hospitality the School of Mathematics of Georgia Institute of Technology, 
where part of this work was done.
L. C. was partially supported by NSF grant DMS-1500943 and by Emory University. 
The authors are also especially grateful to an anonymous referee who pointed out a crucial mistake
in the first version of the manuscript.


\section{Set-up}\label{eq}


It is convenient to define 
\begin{equation}\label{fag}
\f:=\w t\,,\qquad\a:=\frac g\w\,,\qquad\g:=\frac h\w\,,
\end{equation}
so that we can rewrite (\ref{eq:Winf}) as
$$
W'_\infty\psi(\x,\f):=\int_{-\infty}^\f \dd \f' J_0\left(\a(\f-\f')\right)e^{\ii \a \x (\f-\f')}V(\f')\hp(\x,\f')\,,
$$
and hence a periodic solution of (\ref{duat}) with frequency $\w$ should satisfy
\begin{equation}\label{dua}
(\uno+\ii\g W'_{\infty})\hp(\x,\f)=1\,.
\end{equation}
Such a solution will be  explicitly constructed. 

\smallskip

Note that the $\inf$ in \eqref{eq:dioph-alpha} is indeed a $\min$, and it is attained either at 
$k=\lfloor2\a\rfloor$, i.e. the integer part of $2\al$, or at $k=\lceil2\a\rceil:=\lfloor 2\al\rfloor+1$. Moreover
$\bar{\upepsilon}<1$.

Recall the formula
\be\label{eq:jk-tau}
j(\tau):=\int_0^\infty dt J_0(t)e^{\ii\t t}=\frac{\chi(|\t|\leq1)+\ii\sign(\t)\chi(|\t|>1)}{\sqrt{|1-\t^2|}}\,.
\ee

The proof of \eqref{eq:jk-tau} can be found for instance in \cite[Lemma A.3]{CG}. Unfortunately in 
\cite[(A.11)]{CG} the $\sign(\t)$ in the imaginary part is mistakenly omitted, whereas
it is clear from the proof that
it should appear; see also \cite[(A.12)]{CG}.

Set
\begin{equation}\label{jk}
j_k(\x) :=\frac{1}{\al}j(\x + \frac{k}{\al})
\end{equation}
and let us define $\xi_0:=1$, $\xi_0^*:=-1$ and for $k\in\N$
\be
\xi_k:=\sign k-\frac k\a\,,\quad\xi^*_k:=\sign k+\frac k\a\,. 
\ee

\begin{lemma}\label{lemma:jk-crux}
For all $k\neq0$ one has
\be\label{eq:jk-crux}
j_k(\xi)=\frac{\chi(\sign (k)(\xi-\xi_k)\leq0)+\ii\sign (k)\chi(\sign (k)(\xi-\xi_k)>0)}{\a\sqrt{|(\xi-\xi_k)(\xi+\xi^*_k)|}}\,. 
\ee
\end{lemma}

\begin{proof}
Using (\ref{eq:jk-tau}) and \eqref{jk} we have
\be\nonumber
j_k(\xi)=\frac{\chi(|\a\xi+k|\leq\a)+\ii\sign(\a\xi+k)\chi(|\a\xi+k|>\a)}{\a\sqrt{|(\xi-\xi_k)(\xi+\xi^*_k)|}}\,. 
\ee

Let us write
\[
\chi(|\a\xi+k|\leq\a)=\chi(\xi\leq1-\frac k\a)\chi(\xi\geq-1-\frac k\a)\nn
\]
and
\[
\begin{aligned}
\sign(\a\xi+k)\chi(|\a\xi+k|>\a)&=\chi(\a\xi+k>\a)-\chi(\a\xi+k<-\a)\nn\\
&=\chi(\xi>1-\frac k\a)-\chi(\xi<-1-\frac k\a)\,.\nn
\end{aligned}
\]
Now we note that since $\xi\in[-1,1]$, if $k\geq1$ then $\chi(\xi<-1-\frac k\a)=0$ and if $k\leq-1$ then $\chi(\xi>1-\frac k\a)=0$. This implies
\[
\chi(\xi\leq1-\frac k\a)\chi(\xi\geq-1-\frac k\a)=\chi(\sign (k)(\xi-\xi_k)\leq0)\,,\nn
\]
and
\[
\chi(\xi>1-\frac k\a)-\chi(\xi<-1-\frac k\a)=\sign (k)\chi(\sign (k)(\xi-\xi_k)>0)\,
\]
so that the assertion is proven. 
\end{proof}

Note that by the Lemma \ref{lemma:jk-crux}, $j_k(\xi)$ is either real or purely imaginary.
On the other hand $j_0(\x)$ is always real, while $j_k(\x)$ is purely imaginary for $|k|>2\a$.

We conveniently localise the functions $j_k$ about their singularities. Let $r>0$ and set 
\be\label{eq:jloc}
\begin{cases}
Lj_0(\xi)&:=j_0(\xi)(\chi(\x<-1+r)+\chi(\x>1-r))\,,\\
 Lj_k(\x) &:=j_k(\x)\chi(|\x-\x_k|<r)\,,\\
Rj_{k}(\x)&:=j_k(\x) - Lj_k(\x)\,,\qquad k\in\ZZZ\setminus\{0\}.
\end{cases}
\ee

The following properties are proved by straightforward computations.

\begin{lemma}\label{lemma:easy}
\hspace{1cm}
\begin{itemize}
\item [i)] $\xi_k=-\xi_{-k}$ and $\xi^*_k=-\xi^*_{-k}$;
\item [ii)] 
One has
$$
\min_{\substack{|k|,|k'|\le \lfloor2\al\rfloor \\ k\ne k'}} |\x_k - \x_{k'}|=\frac{\bar{\upepsilon}}{\a}\,;
$$
\item [iii)] $\xi_k>0$ if and only if $k<-\a$ or $0<k<\a$. 
\item[iv)] One has
\be\nonumber
\xi_k>\x_{k'} \Longleftrightarrow 
\begin{cases}
k'>k>0\\
k<k'<0\\
k'>0, k<0, k'-k>2\a\\
k'<0, k>0, k'-k>-2\a\\
\end{cases}
\ee
\item[v)] If $|k|>2\a$ then $|\xi_k|>1$;  
\item[vi)]  For $k\geq1$ and $r\in(0,(4\a)^{-1})$ one has
\bea
Lj_k(\xi)&=&\frac{\chi(\xi_k-r<\xi\leq\xi_k)+\ii\chi(\xi_k<\xi<\xi_k+r)}{\a\sqrt{|(\xi-\xi_k)(\xi+\xi^*_k)|}}\label{eq:Lj+}\\
Rj_k(\xi)&=&\frac{\chi(\xi\leq\xi_k-r)+\ii\chi(\xi\geq\xi_k+r)}{\a\sqrt{|(\xi-\xi_k)(\xi+\xi^*_k)|}}\label{eq:Rj+}\,
\eea
and for $k\leq-1$
\bea
Lj_k(\xi)&=&\frac{\chi(\xi_k<\xi <\xi_k+r)-\ii\chi(\xi_k-r<\xi\leq\xi_k)}{\a\sqrt{|(\xi-\xi_k)(\xi+\xi^*_k)|}}\label{eq:Lj-}\\
Rj_k(\xi)&=&\frac{\chi(\xi\geq\xi_k+r)-\ii\chi(\xi\leq\xi_k-r)}{\a\sqrt{|(\xi-\xi_k)(\xi+\xi^*_k)|}}\label{eq:Rj-}\,. 
\eea
\item[vii)] There exist $c_1,c_2>0$ such that for all $k\in\Z$
\be\label{eq:chisarebbeZ}
\inf_{\xi\in[-1,1]}|Lj_k(\xi)|\geq \frac{c_1}{\a\sqrt r}\,,\quad \sup_{\xi\in[-1,1]}|Rj_k(\xi)|\leq \frac{c_2}{\a\sqrt r}\,. 
\ee
\item[viii)] For $|k|>2\a$ and $\upepsilon<\bar\upepsilon$ one has
\begin{equation}\label{jbello}
|j_k(\x)|\le \frac{c_0}{\sqrt{\a{\upepsilon}}}\,,
\end{equation}

\end{itemize}
\end{lemma}

Fix $\upepsilon<\bar{\upepsilon}$  (say $\upepsilon=\bar{\upepsilon}/2$)
and take $r\in(0,r^*)$, with $r^*<\frac{\upepsilon}{4\a}$ 
so that in particular property (vi)
in Lemma \ref{lemma:easy} above is satisfied and moreover one has
\begin{equation}\label{separati}
Lj_k(\x) Lj_{k'}(\x) = 0 \qquad \mbox{for }\quad k\ne k'
\end{equation}
by property (ii) of Lemma \ref{lemma:easy}.


\

Combining a (formal) expansion as a power series in $\g$ and Fourier series in $\f$ (i.e. the so-called
Lindstedt series) we can now obtain a formal series representation for the solution 
of (\ref{dua}) which is the starting point of our analysis. Precisely, we start by writing
\begin{equation}\label{taylor}
\psi(\x,\f) = \sum_{n\ge0}\g^n \psi_{n}(\x,\f)\,,
\end{equation}
so that inserting \eqref{taylor} into \eqref{dua} we see that the coefficients $\psi_n$ must satisfy
\be
\psi_0=1\,,\quad \psi_n=-\ii W'_\infty[\psi_{n-1}]\,.
\ee
We now expand
\[
\psi_n(\x,\f)=\sum_{k\in\ZZZ}\psi_{n,k}(\x) e^{\ii k\f}.
\]

Using that
\be\label{eq:Winf-conv}
(W_\infty \psi_n)_k(\x)=j_k(\x)\sum_{\mu\in\ZZZ}V_{k-\mu}\psi_{n,\mu}(\x)\,,
\ee
by a direct computation we obtain
\be\label{eq:yk}
\left\{
\begin{aligned}
\psi_1(\x,\f)&=\sum_{k_1\in \ZZZ}j_{k_1}(\x)V_{k_1}e^{\ii k_1\f}\,,\\
\psi_2(\x,\f)&=\sum_{k_1,k_2\in\ZZZ} j_{k_1+k_2}(\x)V_{k_2} j_{k_1}(\x)V_{k_1} e^{\ii(k_1+k_2)\f}\\
 &\vdots\\
\psi_n(\f)&=\sum_{k_1,\dots,k_n\in \ZZZ}\Big(\prod_{i=1}^n  j_{\mu_i}(\x) V_{k_i}\Big)e^{\ii\mu_n\f}\,,
\end{aligned}
\right.
\ee
where we denoted
\be\label{conserva}
\mu_p=\mu(k_1,\ldots,k_p):=\sum_{j=1}^p k_j\,.
\ee
Therefore we arrive to write the formal series
\be\label{formale}
\begin{aligned}
\tilde \psi(\x,\f;\g):&=\sum_{\mu\in\ZZZ}e^{\ii\mu\f}\psi_\mu(\x;\g)
=\sum_{\mu\in\ZZZ}e^{\ii\mu\f} \sum_{N\geq0}(-\ii\g)^N
\psi_{N,\mu}(\x)\\
&=\sum_{\mu\in\ZZZ}e^{\ii\mu\f}\sum_{N\geq0}
\sum_{\substack{k_1,\ldots,k_N\in \ZZZ\\ \mu_N=\mu}}
(-\ii\g)^{N}
\Big(\prod_{p=1}^N j_{\mu_p}(\x) V_{k_p}\Big)\,,
\end{aligned}
\ee
which solves\eqref{duat} to all orders in $\g$. Note that for each $N\in\N$ the coefficient of $\g^N$ 
is a sum of singular terms.  This makes it difficult (if not impossible) to show 
 the convergence of (\ref{formale}), and we will instead 
prove the convergence of a resummed series which solves the equation.

\section{Proof of the Theorem}\label{renorm}

To explain our construction of the series giving a solution of \eqref{duat}, it is useful to introduce a slightly modified version of the 
graphical formalism of \cite{CG}, inspired by the one developed in the context of KAM theory (for a review see for instance \cite{G10}).

Since our problem is linear, we shall deal with linear trees, or \emph{reeds}.
Precisely,
an oriented tree is a finite graph with no cycle, such that all the lines are oriented toward a single point 
(the \emph{root})
which has only one incident line (called \emph{root line}). All the points in a tree except the root are called
\emph{nodes}. 
Note that in a tree the orientation induces a natural total ordering ($\preceq$) on the set of the nodes $N(\rho)$
and lines.
If a vertex $v$ is attached to a line $\ell$ we say that $\ell$ exits $v$ if $v\preceq\ell$, otherwise we say that 
$\ell$ enters $v$.
Moreover, since a line $\ell$ may be identified by the node $v$ which it exits, we have a natural total ordering
also on the set of lines $L(\rho)$. We  call \emph{end-node} a node with no line entering it, and
\emph{internal node} any other node. We say that a node has \emph{degree} $d$ if it has exactly $d$ incident
lines. Of course an end-node has degree one.
We call \emph{reed} a labelled rooted tree in which each internal node has degree two.

Given a reed $\rho$ we associate labels with each node and line as follows.
We associate with each node $v$ a \emph{mode label} $k_v\in \ZZZ$ and with each line 
$\ell$ a {\it momentum}
$\mu_\ell \in \ZZZ$ with the constraint
\begin{equation}\label{conservareed}
\mu_\ell = \sum_{v\prec \ell} k_v\,.
\end{equation}
Note that \eqref{conservareed} above is a reformulation of \eqref{conserva}.
We call \emph{order} of a reed $\rho$ the number $\# N(\rho)$ and \emph{total momentum} of a 
reed the momentum associated with the root line.

$\Theta_{N,\mu}$ denotes the set of reeds of order $N$ and total momentum $\mu$. 
We say that a line $\ell$ is \emph{regular} if 
$|\mu_\ell|\geq \lceil2\al\rceil$, otherwise we say it is \emph{singular}.
With every singular line $\ell$ we attach a 
further {\it operator label} $\calO_\ell\in\{L,R\}$; if $\ell$ is singular 
we say that it is \emph{localised} if  $\calO_\ell=L$, otherwise we say that it is {\emph{regularised}}.

We then associate with each node $v$ a \emph{node factor}
\be\label{nodefactor}
\calF_v = V_{k_v}
\ee
and with each line $\ell$ a \emph{propagator}
\be\label{propagator}
\calG_\ell(\x) = \left\{
\begin{aligned}
&j_{\mu_\ell}(\x)\,,\qquad \ell\mbox{ is regular}\\
&\calO_\ell j_{\mu_\ell}(\x)\,,\qquad \ell \mbox{ is singular},
\end{aligned}
\right.
\ee
so that we can associate with each reed $\rho$ a value as
\be\label{val}
\Val(\rho) = \Big(\prod_{v\in N(\rho)} \calF_v\Big) \Big(\prod_{\ell\in L(\rho)} \calG_\ell(\x)\Big).
\ee
In particular one has formally
\be\label{ovvio}
\psi_{N,\mu} = \sum_{\rho\in \Theta_{N,\mu}} \Val(\rho) \,.
\ee

\begin{rmk}\label{stoqua}
If in a reed $\rho$ with $\Val(\rho)\ne0$ there is a localised line $\ell$, i.e. if $\calO_\ell=L$,
then all the lines with momentum  $\mu\ne \mu_\ell$ are either regular or
regularised. Indeed if $\ell$ is localised, then by \eqref{eq:jloc} we have that
$\x$ is $r$-close to $\x_{\mu_\ell}$ and hence it cannot be $r$-close to $\x_\mu$
for $\mu\ne \mu_\ell$; see also \eqref{separati}.
\end{rmk}

Given a reed $\r$ we say that a connected subset $\mathtt{s}$ of nodes and lines in $\rho$ is a 
{\it closed-subgraph} if $\ell\in L(\mathtt{s})$ implies that
$v,w\in N(\mathtt{s})$ where $v,w$ are the nodes $\ell$ exits and enters respectively.
We say that a closed-subgraph $\mathtt{s}$ has degree $d:=|N(\mathtt{s})|$. 
We say that a line $\ell$ {\it exits} a closed-subgraph ${\mathtt{s}}$ if it exits a node in
$N(\mathtt{s})$ and enters either the root (so that $\ell$ is the root line) or a node
in $N(\rho)\setminus N(\mathtt{s})$. Similarly we say that a line {\it enters} $\mathtt{s}$
if it enters a node in $N(\mathtt{s})$ and exits a node in $N(\rho)\setminus N(\mathtt{s})$.
We say that a closed-subgraph $\mathtt{s}$ is a {\it resonance} if it has an
exiting line $\ell_{\mathtt{s}}$ and an entering line
$\ell_{\mathtt{s}}'$, both $\ell_{\mathtt{s}}$ and $\ell_{\mathtt{s}}'$ are localised
(so that in particular
by Remark \ref{stoqua}  the exiting and entering lines of a resonance must  carry the same momentum),
while all lines $\ell\in L(\mathtt{s})$ have momentum $\mu_{\ell}\ne\mu_{\ell_{\mathtt{s}}} $.

Note that by \eqref{conservareed}
one has
\begin{equation}\label{sec}
\sum_{v\in N(\mathtt{s})} k_v=0\,,
\end{equation}

We denote by $\TT_{d,\mu}$ the set of resonances with degree $d$ and entering and exiting
lines with momentum $\mu$.  Note that if $d=1$ then $\TT_{1,\mu}$
is constituted by a single node $v$ with mode $k_v=0$.

Let us set
 \begin{equation}
 \MM_{d,\mu}(\x) := \sum_{\mathtt{s}\in\TT_{d,\mu}}\Val(\mathtt{s})\,,
 \end{equation}
 where  we define the value of a resonance $\mathtt{s}$ as in \eqref{val} but with the products
restricted to nodes and lines in $\mathtt{s}$, namely
 $$
 \Val(\mathtt{s}) := \Big(\prod_{v\in N(\mathtt{s})} \calF_v\Big) \Big(\prod_{\ell\in L(\mathtt{s})} 
 Rj_{\mu_\ell}(\x)\Big)\,.
 $$

Next we proceed with the proof, which we divide into several steps.

\begin{proof}[Step 1: resummation.]
 The idea behind resummation can be roughly described as follows. The divergence of the sum
in \eqref{ovvio} is due to the presence of localised lines (and their possible accumulation). If a
reed $\rho_0\in \Theta_{N,\mu}$ has a localised line $\ell$, say exiting a node $v$, then we can
consider another
reed $\rho_1\in\Theta_{N+1,\mu}$ obtained from $\rho_0$ by inserting an extra node $v_1$
with $k_{v_1}=0$ and an extra localised line $\ell'$ between $\ell$ and $v$, i.e. $\rho_1$ has
an extra resonace of degree one. Of course, while
$\rho_0$ is a contribution to $\psi_N(\f)$, $\rho_1$ is a contribution to $\psi_{N+1}(\f)$,
so when (formally) considering the whole sum, the value of $\rho_1$ will have
an extra factor $(-\ii\g)$. In other words, in the formal sum \eqref{formale} there will be a term of the form
\[
\begin{aligned}
\Val(\rho_0)+(-\ii\g) \Val(\rho_1) &= (\mbox{common factor} )\big(Lj_{\mu_\ell}(\x) +
 Lj_{\mu_\ell}(\x) (-\ii\g) V_0 Lj_{\mu_\ell}(\x)\big) \\
 &= (\mbox{common factor} ) Lj_{\mu_\ell}(\x)\big(1 +
(-\ii\g) V_0 Lj_{\mu_\ell}(\x)\big)
 \end{aligned}
\]
Of course we can indeed insert any chain of resonances of degree one, say of length $p$,
so as to obtain a reed $\rho_p\in\Theta_{N+p,\mu}$,
and when summing their values together we formally have
\[
\begin{aligned}
\sum_{p\ge0} (-\ii\g)^p\Val(\rho_p) &=(\mbox{common factor} ) Lj_{\mu_\ell}(\x)\big( 1+
(-\ii\g) V_0 Lj_{\mu_\ell}(\x) \\
&\qquad\qquad \qquad\qquad+ (-\ii\g )V_0 Lj_{\mu_\ell}(\x)
 (-\ii\g) V_0 Lj_{\mu_\ell}(\x) +\ldots \big)\\
 &=(\mbox{common factor} ) Lj_{\mu_\ell}(\x)\sum_{p\ge0 } ((-\ii\g) V_0 Lj_{\mu_\ell}(\x))^p\\
 &=(\mbox{common factor} ) \frac{Lj_{\mu}(x)}{1+ \ii\g V_0 Lj_\m(\x)}
\end{aligned}
\]
In other words we  formally replace the sum over $N$ of the sum of reeds in $\Theta_{N,\mu}$
with the sum of reeds where no resonance of degree one appear, but with the localised propagators
replaced with 
\[
\frac{Lj_{\mu}(x)}{1+ \ii\g V_0 Lj_\m(\x)}.
\]

Clearly in principle we can perform this formal substitution considering resonances of any degree.
Here it is enough to consider resummations only of resonances of degree one and two. 
The advantage of such a formal
procedure is that the localised propagators do not appear anymore.
However, since the procedure is only formal,
 one has to prove not only that the new formally defined object is indeed well defined, but also
that it solves \eqref{dua}.

Having this in mind,
let $\Theta^{\RR}_{N,\mu}$ be the set of reeds in which no resonance of deree $1$ nor $2$ appear, 
and define
\bea
\MM_\mu(\x)=\MM_\mu(\x,\g) &:=&(-i\g) \MM_{0,\mu}(\x) + (-i\g)^2  \MM_{1,\mu}(\x)\nn\\
&=&-i\g V_0 -\g^2\sum_{k\in\ZZZ} V_k Rj_{k+\mu}(\x) V_{-k} \,.\label{emme}
\eea

In Section \ref{importa} we  prove the following result.

\begin{prop}\label{prop:main}
For all $\mu\in\Z\cap[-2\a,2\a]$ and for
\begin{equation}\label{enorme}
\g\in\left\{
\begin{aligned}
&(0,+\io)\quad\qquad V_0\ge 0 \\
&(0,  c \sqrt{\frac{\upepsilon}{\al}}\frac{|V_0|}{\|V\|_{L^2}^2})
\quad\qquad V_0< 0
\end{aligned}
\right.
\end{equation}
where $c$ is a suitable absolute constant, one has
\be
\inf_{\xi\in[-1,1]} |1-\mathcal M_\mu(\xi)Lj_\mu(\xi)|\geq\frac12\,.  
\ee
\end{prop}

Proposition \ref{prop:main} allows us to set
\begin{equation}\label{staqua}
Lj^\RR_\mu(\x):=\frac{Lj_\mu(\x)}{1-\MM_\mu(\x,\g) Lj_\mu(\x)}\,.
\end{equation}

For any $\r\in\Theta^{\RR}_{N,\mu}$ let us define the renormalised value of $\rho$ as
\be\label{renval}
\Val^{\RR}(\rho):=\left(\prod_{v\in N(\r)} \calF_{v}\right)\left(\prod_{\ell\in L(\r)}\calG^{\RR}_{\ell}\right)\,,
\ee
where
\be\label{renprop}
\calG^{\RR}_{\ell_i}=\left\{
\begin{aligned}
&Lj_\mu^{\RR}(\x),\qquad \qquad |\mu_{\ell_i}|\le \lfloor2\al\rfloor,\quad \calO_{\ell_i}=L\,,\\
&Rj_\mu(\x),\qquad \qquad |\mu_{\ell_i}|\le \lfloor2\al\rfloor,\quad \calO_{\ell_i}=R\,,\\
&j_{\mu_{\ell_i}}(\x),\qquad\qquad |\mu_{\ell_i}|\geq \lceil2\al\rceil\,.
\end{aligned}
\right.
\ee

In particular if $\lfloor2\al\rfloor=0$ we have to renormalise only $j_0$, which is the case in 
our previous paper \cite{CG}. Then we define
\be\label{coeffrin}
\psi_\mu^\RR(\x;\g):=\sum_{N\ge1}(-\ii\g)^N
\sum_{\r\in\Theta^\RR_{N,\mu}}\Val^\RR(\rho)\,,
\ee
so that
\begin{equation}\label{asy1}
\psi^\RR(\f;\x,\g):=
\sum_{\mu\in\ZZZ}e^{\ii\mu\f}\psi_\mu^\RR(\x;\g)\,,
\end{equation}
is the renormalised series we want to prove to be a regular solution of (\ref{dua}).
\end{proof}

\begin{proof}[Step 2: radius of convergence.]

First of all we prove that the function \eqref{asy1} is well defined.

We start by noting that the node factors are easily bounded by \eqref{eq:V}.
The propagators defined in (\ref{renprop}) are bounded as follows. If $|\mu_\ell|\ge \lceil2\al\rceil$
formula (\ref{jbello}) yields
$$
|j_\mu(\x)|\leq\frac{c_0}{\sqrt{2\lceil2\al\rceil\upepsilon}}\,,
$$
while for $|\mu|\le\lfloor2\a\rfloor$, by (\ref{eq:chisarebbeZ}) we have 
$$
|{Rj}_\mu(\x)|\le \frac{c_2}{\a\sqrt r}\,.
$$

Regarding the resummed propagators the bound is more delicate. We start by denoting
\begin{equation}\label{minimo}
\overline{V}:=
\left\{
\begin{aligned}
&0\qquad \mbox{ if } V_k=0,\ \forall \,k\ge1\\
&\max_{k\in\ZZZ\setminus\{0\} }|V_k|^2 \quad \mbox{otherwise}\,,
\end{aligned}
\right.
\end{equation}
and
\be\label{vbarl}
\underline{V}_{\leq 2\a}:=\left\{
\begin{aligned}
&0\qquad \mbox{ if } V_k=0,\ \forall \,k=1,\ldots,\lfloor2\al\rfloor \\
&\min_{\substack{|k|\leq\lfloor2\a\rfloor \\ V_k\ne0}}|V_k| \quad \mbox{otherwise}\,,
\end{aligned}
\right.
\qquad 
\overline{V}_{> 2\a}:=
\left\{
\begin{aligned}
&0\qquad \mbox{ if } V_k=0,\ \forall \,k\ge \lceil2\al\rceil \\
&\max_{\substack{|k|>\lceil2\a\rceil \\ V_k\ne0}}|V_k| \quad \mbox{otherwise}\,,
\end{aligned}
\right.
\ee

In Section \ref{stimazze} we prove the following result.

\begin{prop}\label{lemma:jR}
There is a constant $c>0$ such that
\be\label{stastimazza}
|Lj^\RR_\mu(\x)|\leq T(V,\upepsilon,\a;\g)=T(\g):=
\begin{cases}
\frac{c}{\g|V_0|}&\mbox{if}\quad V_0\neq0\mbox{  and  } \g \leq c\sqrt{\frac{\upepsilon}{{\al}}}
\frac{|V_0|}{\|V\|^2_{L^2}}\,;\\
c\sqrt{\frac{\a}{\upepsilon}}\g^{-2} \underline{V}^{-2}_{\le 2\al}&\mbox{if}\quad V_0=0,\
\underline{V}_{\le 2\al}\ne0\,;\\
c{{\sqrt{\a}}{}}\g^{-2} \ol{V}^{-2}_{> 2\al}&\mbox{if}\quad V_0=0,\
\underline{V}_{\le 2\al}=0\,.
\end{cases}
\ee
\end{prop}

Let us set now
\be\label{eq:B}
B=B(r,\a,\upepsilon):=\max\left(\frac{1}{\a\sqrt r},\frac{1}{\sqrt{\a\upepsilon}}\right) =\frac{1}{\al\sqrt{r}}
 \,,\qquad
C_1:=\max(c_0,c_2/2)\,.
\ee
Note that if in a reed $\rho\in\Theta^\RR_{N,\mu}$ there are $l$ localised lines, 
we have
\begin{equation}\label{mestavoascorda1}
\begin{aligned}
|\Val^\RR(\r)|&=\Big(\prod_{v\in N(\rho)} |\calF_v|\Big) \Big(\prod_{\ell\in L(\rho)}| \calG_\ell|\Big)\\
&\le \Big( C_0e^{-\s\sum_{v\in N(\rho)}|n_v|}\Big) \Big(\prod_{\ell\in L(\rho)}| \calG_\ell|\Big)\\
&\le
C_0C_1^NB^{N}T(\g)^{l }e^{-\s|\mu|}\,,
\end{aligned}
\end{equation}
for some constant $C_0>0$.
By construction, in a renormalised reed  there must
be at least two lines between two localised lines, since we resummed the resonances
of degree one and two. This implies that a 
reed with $N$ nodes can have at most $l=\lceil N/3\rceil$ localised lines.

Then by (\ref{coeffrin}) we obtain
\be\label{stimatotale1}
|\psi_\mu^\RR(\x;\g)|\le C \sum_{N\ge 1}\g^{N} B^{N}T(\g)^{\frac N3}e^{-\s|\mu|/2}\,,
\ee
so that
the series above converge for
\be\label{eq:cond2gamma}
\g^3 T(\g)B^3<1\,. 
\ee
This entails
\be\label{eq:conv-gamma}
\g<
\begin{cases}
\min\left(\sqrt[3]{|V_0|B^{-3}},c\sqrt{\frac{\upepsilon}{{\al}}}
\frac{|V_0|}{\|V\|^2_{L^2}}\right)&V_0\neq0\,;\\
c^{-1} B^{-3}\sqrt{\frac{\upepsilon}{\al}}\ \underline{V}_{\le2\al}^2 &V_0=0\,,\underline{V}_{\le2\al}>0;\\
c^{-1} B^{-3}{\frac{1}{\sqrt{\al}}} \ol{V}_{>2\al}^2
&V_0=0\,,\underline{V}_{\le2\al}=0.\\
\end{cases}
\ee

Therefore under such smallness condition on $\g$, the function $\psi^\RR(\f;\x,\g)$ (recall (\ref{asy1})) is 
analytic w.r.t. $\f\in\TTT$, uniformly in $\x\in[-1,1]$ and for $\g$ small enough  
\end{proof}


Choosing $\upepsilon=\ol{\upepsilon}/2$ and
 $r:=\frac{\upepsilon}{8\a}$ we 
have by (\ref{eq:B})
$$
B(r,\a,\upepsilon)^{-3}=\sqrt{\frac{{\al^3\upepsilon^3}}{{64}}}\,,
$$
so that condition \eqref{eq:conv-gamma} implies that
 the series converges for $\g\le \g_0:=c_1\g_1$ where
\begin{equation}\label{gamma0}
\g_1:=\begin{cases}
\min\left(\sqrt[3]{|V_0|(\frac{\sqrt{\al\bar{\upepsilon}}}{4})^{3}},\sqrt{\frac{\bar \upepsilon}{{2\al}}}
\frac{|V_0|}{\|V\|^2_{L^2}}\right)&V_0\neq0\,;\\
 (\frac{\sqrt{\al\bar{\upepsilon}}}{4})^{3}\sqrt{\frac{\bar\upepsilon}{2\al}}\ 
\underline{V}_{\le2\al}^2 &V_0=0\,,\underline{V}_{\le2\al}>0;\\
 (\frac{\sqrt{\al\bar{\upepsilon}}}{4})^{3}{\frac{1}{\sqrt{\al}}} \ol{V}_{>2\al}^2
&V_0=0\,,\underline{V}_{\le2\al}=0.\\
\end{cases}
\end{equation}
and $c_1:=\min\{c,c^{-1}\}$.


\begin{proof}[Step 3: $\psi^\RR(\f;\x,\g)$ solves \eqref{dua}.]
Now we want to prove that
$$
(\uno+i\g W'_{\infty})\psi^{\RR}(\f;\x,\g)=1\,.
$$
This is essentially a standard computation.
Using (\ref{coeffrin}) and (\ref{asy1}), the last equation can be rewritten as
\be\label{eq:W'-step5}
i\g W'_{\infty}\psi^{\RR}(\f;\x,\g)=1-\psi^{\RR}(\f;\x,\g)=-\sum_{\mu\in\Z}e^{\ii \mu \f}\sum_{N\geq1}(-\ii\g)^N
\sum_{\r\in\Theta^\RR_{N,\mu}}\Val^\RR(\rho)\,.
\ee
Moreover thanks to (\ref{eq:Winf-conv}) we can compute
\bea
i\g W'_{\infty}\psi^{\RR}(\f;\x,\g)&=&\ii\g\sum_{\mu\in\Z}\psi^\RR(\x;\g)(W'_\infty e^{\ii\mu\f})\nn\\
&=&\ii\g\sum_{\mu\in\Z} e^{\ii\mu\f}j_{\mu}(\x)\sum_{k\in\Z} V_{\mu-k}\psi_k^{\RR}(\x;\g)\nn\\
&=&\ii\g\sum_{\mu\in\Z} e^{\ii\mu\f}\sum_{N\ge0}(-\ii\g)^{N}j_{\mu}(\x)\sum_{\mu_1+\mu_2=\mu}V_{\mu_1}
\sum_{\r\in\Theta^\RR_{N,\mu_2}}\!\!\!\!
\Val^\RR(\rho)\,.\nn
\eea
Thus we can write (\ref{eq:W'-step5}) in terms of Fourier coefficients as
\be\label{perico}
\sum_{N\ge1}(-\ii\g)^{N}j_{\mu}(\x)\sum_{\mu_1+\mu_2=\mu}V_{\mu_1}
\sum_{\r\in\Theta^\RR_{N-1,\mu_2}}\!\!\!\!
\Val^\RR(\rho)=\sum_{N\geq1}(-\ii\g)^N
\sum_{\r\in\Theta^\RR_{N,\mu}}\Val^\RR(\rho)\,.
\ee

Note that the root line $\ell$ of a reed has to be renormalised only if it carries momentum 
label $|\mu_\ell|\le\lfloor{2\al}\rfloor$ and operator $\calO_\ell=L$, thus for $|\mu_\ell|\ge\lceil{2\al}\rceil$, or
 $|\mu_\ell|\le\lfloor2\al\rfloor$ and $\calO_\ell=R$ we see immediately that \eqref{perico} holds.

Concerning the case $\mu_\ell=\mu$ with $|\mu|\le\lfloor2\al\rfloor$ and $\calO_\ell=L$, we first note that
$$
j_{\mu}(\x)
\sum_{\mu_1+\mu_2=\mu}V_{\mu_1}
\sum_{\r\in\Theta^\RR_{N-1,\mu_2}}\!\!\!\!
\Val^\RR(\rho)=
\sum_{\r\in\ol{\Theta}^\RR_{N,\mu}}
\Val^\RR(\rho)
$$
where $\ol{\Theta}^\RR_{N,\mu}$ is the set of reeds such that the root line may exits a resonance of degree 
$\le2$,
so that equation \eqref{perico} reads
\be\label{sforzo}
\psi^\RR_\mu(\x;\g)= \sum_{N\ge1}(-\ii\g)^{N}
\sum_{\r\in\ol{\Theta}^\RR_{N,\mu}}
\Val^\RR(\rho)\,,
\ee

Let us split
\be\label{split}
\ol{\Theta}^\RR_{N,\mu}=\widetilde{\Theta}^\RR_{N,\mu}\cup \hat{\Theta}^\RR_{N,\mu}\,,
\ee
where  
$\hat{\Theta}^\RR_{N,\mu}$ are the reeds such that the root line indeed exits a resonance of degree $\le2$,
while $\widetilde{\Theta}^\RR_{N,\mu}$ is the set of all other renormalised reeds.
Therefore we have
\be\label{pezzouno}
\sum_{N\ge1}(-\ii\g)^N\sum_{\rho\in\widetilde{\Theta}^\RR_{N,\mu}} \Val^\RR(\rho) = 
Lj_\mu^\RR(\x) \sum_{\mu_1+\mu_2=\mu} (\ii\g V_{\mu_1}) \psi_{\mu_2}^\RR(\x;\g)
\ee
and
\be\label{pezzodue}
\sum_{N\ge1}(-\ii\g)^N\sum_{\rho\in\hat{\Theta}^\RR_{N,\mu}} \Val^\RR(\rho) =
 Lj_\mu^\RR(\x)\MM_{\mu}(\x,\g)Lj_\mu^\RR(\x) \sum_{\mu_1+\mu_2=\mu} (\ii\g V_{\mu_1}) 
 \psi_{\mu_2}^\RR(\x;\g)\,,
\ee
so that summing together \eqref{pezzouno} and \eqref{pezzodue} we obtain $\psi^\RR_\mu(\x;\g)$.
\end{proof}
This concludes the proof of the Theorem.


\section{Proof of Proposition \ref{prop:main}}\label{importa}

In this section we prove Proposition \ref{prop:main}.
We will consider explicitly the case $\mu\in\N$, since negative $\mu$ are dealt with in
a similar way.

Set for brevity
\bea
D_k(\xi)&:=&\a\sqrt{|(\xi-\xi_k)(\xi+\xi^*_k)|}\,,\label{eq_D}\\
A_{\mu,k}(\xi)&:=&-D^{-1}_{\mu+k}(\xi)+D^{-1}_{\mu-k}(\xi)\,,\label{eq:A}\\
G_{\mu,k}(\xi)&:=&(Rj_{\mu+k}(\xi)+Rj_{\mu-k}(\xi))Lj_\mu(\xi)\,, \label{eq_G}
\eea
and note that we can write
\be\label{eq:MeG}
\mathcal M_\mu(\xi) Lj_\mu(\xi)=-\ii\g V_0 Lj_\mu(\xi) -\g^2\sum_{k\geq1}|V_k|^2G_{\mu,k}(\xi)\,.
\ee

The next two lemmas establish useful properties of the functions $G_{\mu,k}(\xi)$ and $A_{\mu,k}(\xi)$.

\begin{lemma}\label{lemma:Aposi}
Let $k\in\N$ and $r$ sufficiently small. One has
\be\label{eq:Aposi}
\inf_{\xi\in[\xi_\mu-r,\xi_\mu+r]} A_{\mu,k}(\xi)>0\,
\ee
\end{lemma}
\begin{proof}
By explicit calculation
\be\label{eq:Aximu}
A_{\mu,k}(\xi_\mu)=
\begin{cases}
\frac{2\sqrt k}{\sqrt{4\a^2-k^2}(\sqrt{k(2\a-k)}+\sqrt{k(k+2\a)})}>0&k<2\a\,,\\
\frac{4\a}{\sqrt{k^2-4\a^2}(\sqrt{k(k-2\a)}+\sqrt{k(k+2\a)})}>0&k>2\a\,
\end{cases}
\ee
so we can conclude by continuity.
\end{proof}

\begin{lemma}\label{lemma:G}
If $\xi \in(\xi_\mu,\xi_\mu+r)$ one has 
\bea
\Re(G_{\mu,k}(\x))&=&\begin{cases}
-(D_\mu(\x) D_{\mu+k}(\x))^{-1}&1\leq k\leq \mu\\
-(D_\mu(\x) D_{\mu+k}(\x))^{-1}&\mu+1\leq k\leq \lfloor2\a\rfloor\\
\frac{A_{\mu,k}(\xi)}{D_\mu(\xi)}&k\geq \lceil2\a\rceil\,.
\end{cases}\label{eq:ReGr}\\
\Im(G_{\mu,k}(\x))&=&\begin{cases}
(D_\mu (\x)D_{\mu-k}(\x))^{-1}&1\leq k\leq \mu\\
(D_\mu (\x)D_{\mu-k}(\x))^{-1}&\mu+1\leq k\leq \lfloor2\a\rfloor\\
0&k\geq \lceil2\a\rceil\,.\label{eq:ImGr}
\end{cases}
\eea
If $\xi \in(\xi_\mu-r,\xi_\mu]$ one has
\bea
\Re(G_{\mu,k}(\x))&=&\begin{cases}
(D_\mu (\x)D_{\mu-k}(\x))^{-1}&1\leq k\leq \mu\\
(D_\mu (\x)D_{\mu-k}(\x))^{-1}&\mu\leq k\leq \lfloor2\a\rfloor\\
0&k\geq \lceil2\a\rceil\,.
\end{cases}\label{eq:ReGl}\\
\Im(G_{\mu,k}(\x))&=&\begin{cases}
(D_\mu (\x)D_{\mu+k}(\x))^{-1}&1\leq k\leq \mu\\
(D_\mu (\x)D_{\mu+k}(\x))^{-1}&\mu+1\leq k\leq \lfloor2\a\rfloor\\
-\frac{A_{\mu,k}(\xi)}{D_\mu(\xi)}&k\geq \lceil2\a\rceil\,.\end{cases}\label{eq:ImGl}
\eea
\end{lemma}

\begin{proof}
Since we are considering the case $\mu\geq1$, $Lj_\mu(\x)$ is given by (\ref{eq:Lj+}). Our analysis of $G_{\mu,k}(\xi)$ splits in several cases.

\

\textit{i)} $1\leq k\leq \mu$

In this case $\mu-k\geq0$ and $\xi_{\mu+k}<\xi_\mu<\xi_{\mu-k}$. By (\ref{eq:Rj+}) we write
\bea
Rj_{\mu+k}(\xi)+Rj_{\mu-k}(\xi)&=&D_{\mu+k}^{-1}(\x)\chi(\xi<\xi_{\mu+k}-r)+
D^{-1}_{\mu-k}(\x)\chi(\xi<\xi_{\mu-k}-r)\nn\\
&+&\ii D_{\mu+k}^{-1}(\x)\chi(\xi>\xi_{\mu+k}+r)+\ii D_{\mu-k}^{-1}(\x)\chi(\xi>\xi_{\mu-k}+r)\,. 
\eea
A direct computation gives
\bea
D_\mu(\x) G_{k,\mu}(\x)&=&D^{-1}_{\mu-k}(\x)\chi(\xi_\mu-r<\xi\leq\xi_\mu)-
D^{-1}_{\mu+k}(\x)\chi(\xi_\mu<\xi<\xi_\mu+r)\nn\\
&+&\ii(D^{-1}_{\mu+k}(\x)\chi(\xi_\mu-r<\xi\leq\xi_\mu)+D^{-1}_{\mu-k}(\x)\chi(\xi_\mu<\xi<\xi_\mu+r))\,. 
\eea

\

\textit{ii)} $\mu+1\leq k\leq \lfloor2\a\rfloor$

Now $\mu-k<0$ and $\max(\xi_{\mu+k},\xi_{\mu-k})<\xi_\mu$. Therefore by (\ref{eq:Rj+}), (\ref{eq:Rj-})
\bea
Rj_{\mu+k}(\xi)+Rj_{\mu-k}(\xi)&=&D_{\mu+k}^{-1}(\x)\chi(\xi<\xi_{\mu+k}-r)+
D^{-1}_{\mu-k}(\x)\chi(\xi>\xi_{\mu-k}+r)\nn\\
&+&\ii D_{\mu+k}^{-1}(\x)\chi(\xi>\xi_{\mu+k}+r)-\ii D_{\mu-k}^{-1}(\x)\chi(\xi<\xi_{\mu-k}-r)\,.
\eea
Moreover
\bea
D_\mu(\x) G_{k,\mu}(\x)&=&D^{-1}_{\mu-k}(\x)\chi(\xi_\mu-r<\xi\leq\xi_\mu)-
D^{-1}_{\mu+k}(\x)\chi(\xi_\mu<\xi<\xi_\mu+r)\nn\\
&+&\ii(D^{-1}_{\mu+k}(\x)\chi(\xi_\mu-r<\xi\leq\xi_\mu)+D^{-1}_{\mu-k}(\x)\chi(\xi_\mu<\xi<\xi_\mu+r))\,. 
\eea

\

\textit{iii)} $k\geq \lceil2\a\rceil$
We have $\mu-k<0$, $\xi_{\mu+k}<\xi_\mu<\xi_{\mu-k}$ and again
\bea
Rj_{\mu+k}(\xi)+Rj_{\mu-k}(\xi)&=&D_{\mu+k}^{-1}(\x)\chi(\xi<\xi_{\mu+k}-r)+
D^{-1}_{\mu-k}(\x)\chi(\xi>\xi_{\mu-k}+r)\nn\\
&+&\ii D_{\mu+k}^{-1}(\x)\chi(\xi>\xi_{\mu+k}+r)-\ii D_{\mu-k}^{-1}(\x)\chi(\xi<\xi_{\mu-k}-r)\,.
\eea
Therefore
\be
G_{k,\mu}(\x)=-\ii 
A_{\mu,k}(\xi)Lj_\mu(\xi)=\frac{A_{\mu,k}(\xi)}{D_\mu(\xi)}(\chi(\xi_\mu<\xi<\xi_\mu+r)-
\ii\chi(\xi_\mu-r<\xi\leq\xi_\mu))\,. 
\ee
Note that for $\xi \notin(\xi_\mu-r,\xi_\mu+r)$ all $G_{\mu,k}$ are identically zero and by direct inspection we deduce (\ref{eq:ReGr}), (\ref{eq:ImGr}), (\ref{eq:ReGl}) and (\ref{eq:ImGl}).
\end{proof}

Let us introduce the notation
\begin{equation}\label{capponinuovi}
\begin{aligned}
K_1(\xi)&:=\sum_{k=1}^{\lfloor2\a\rfloor} |V_k|^2(D_{\mu-k}(\x))^{-1}\\
K_2(\xi)&:=\sum_{k>2\al} |V_k|^2A_{\mu,k}(\x)\,,\qquad
\tilde{K}_2(\x):=\sum_{k\ge1} |V_k|^2A_{\mu,k}(\x)
\end{aligned}
\end{equation}
Note that $K_1(\x)$ can vanish if and only if $V_k=0$ for all $|k|\le \lfloor 2\al\rfloor$ (that is $\underline V_{\leq 2\a}=0$). Similarly
$K_2(\x)$ can vanish if and only if $V_k=0$ for all $|k|> \lfloor 2\al\rfloor$ (i.e. $\overline V_{>2\a}=0$), while
$\tilde{K}_2(\x)$ can vanish if and only if the forcing is constant in time, i.e. $V(\om t)\equiv V_0$. 

Recall the notations introduced in \eqref{minimo}--\eqref{vbarl}.
We need the following bounds on $K_1(\x)$, $K_2(\x)$ and $\tilde{K}_2(\x)$.

\begin{lemma}\label{lemma:crux}
There is $c>0$ such that for all $\mu\in\Z$ and all $\x\in(\xi_\mu-r,\xi_\mu+r)$ 
if $K_1(\x)\ne0$ then one has
\begin{equation}\label{stimeK1}
c \sqrt{\frac{\upepsilon}{\al}}\ \underline{ V}_{\leq 2\a}^2 \le
K_1(\x)  \le
c\sqrt{\frac{\al}{\upepsilon}}\|V\|^2_{L^2}\,.
\end{equation}
\end{lemma}

\begin{proof}
By the Lipschitz-continuity of $D^{-1}_{\mu-k}(\xi)$ in $(\xi_\mu-r,\xi_\mu+r)$ 
 there is $c_1\geq0$ such that for all $\xi \in(\xi_\mu-r,\xi_\mu+r)$ 
we have
\be\nn
|D^{-1}_{\mu-k}(\xi)-D^{-1}_{\mu-k}(\xi_\mu)|\leq c_1\frac{r}{\sqrt\a}\,.
\ee
Furthermore since $r\in(0,\frac{\upepsilon}{2\a})$ and
$$
D_{\mu-k}(\xi_\mu)={\sqrt{|2\a-k|k}}\,,
$$
we have 
\bea\label{eq:inf1}
\sum_{k=1}^{\lfloor2\a\rfloor} \frac{|V_k|^2}{ D_{\mu-k}(\x)}&\leq& C_1\left(\sup_{k\in\N}|V_k|^2\sum_{k=1}^{\lfloor2\a\rfloor}\frac{1}{\sqrt{k|k-2\a|}}+\frac{r\|V\|^2_{L^2}}{\sqrt \a}\right)\nn\\
&\leq& {C_2}\|V\|_{L^2}^2(\sqrt{\frac{\al}{\upepsilon}} + \frac{\upepsilon}{\al\sqrt{\al}})\nn\\
&\le& C_3 \sqrt{\frac{\al}{\upepsilon}}\|V\|_{L^2}^2
\nn\,.
\eea
for some constants $C_1,C_2,C_3>0$.
Similarly
\bea
\sum_{k=1}^{\lfloor2\a\rfloor} \frac{|V_k|^2}{ D_{\mu-k}(\x)}&\geq& C_1\left(\underline{ V}_{\leq 2\a}^2\sum_{k=1}^{\lfloor2\a\rfloor}\frac{1}{\sqrt{k|k-2\a|}}-\frac{r\|V\|^2_{L^2}}{\sqrt \a}\right)\nn\\
&\geq& C_2\sqrt{\frac{\upepsilon}{\al}}\ \underline{ V}_{\leq 2\a}^2 \nn\,,
\eea
so the assertion follows.
\end{proof}

\begin{lemma}\label{lemma:bound-A}
There is $c>0$ such that for all $\mu\in\Z$ and all $\x\in(\x_\mu -r , \x_\mu+r)$ 
if $K_2(\x)\ne0$ then one has
\begin{equation}\label{stimeK2}
\frac{c}{\sqrt\a}
\overline{V}^2_{>2\a}
\le
K_2(\x) \le
c\sqrt{\frac{\a}{\upepsilon}}
\|V\|^2_{L^2}.
\end{equation}
\end{lemma}

\begin{proof}
We use again the Lipschitz-continuity of $D^{-1}_{\mu\pm k}(\xi)$ in $(\x_\mu-r,\x_\mu+r)$ to obtain
 that there is $c\geq0$ such that for all $\xi \in(\xi_\mu-r,\xi_\mu+r)$ one has
\be\nn
|A_{\mu,k}(\xi)-A_{\mu,k}(\xi_\mu)|\leq c\frac{r}{\sqrt\a}\,,
\ee
and hence for all $\xi \in(\xi_\mu-r,\xi_\mu+r)$
\be\label{combino}
|\sum_{k>2\a} |V_k|^2 A_{\mu,k}(\xi)-\sum_{k>2\a} |V_k|^2 A_{\mu,k}(\xi_\mu)|\leq c\|V\|_{L^2}\frac{r}{\sqrt\a}\,. 
\ee

Now we have by \eqref{eq:Aximu}
\be\nn
\sum_{k>2\a} |V_k|^2 A_{\mu,k}(\xi_\mu)\leq \|V\|^2_{L^2}\sum_{k>2\a}
\frac{4\a}{\sqrt{k^2-4\a^2}(\sqrt{k(k-2\a)}+\sqrt{k(k+2\a)})}\leq c_1\sqrt{\frac{{\al}}{\upepsilon}}\|V\|^2_{L^2}\,,
\ee
and similarly, if $\hat{k}$ denotes the Fourier mode at which the max in \eqref{vbarl} is attained, we have
\be\nn
\sum_{k>2\a} |V_k|^2 A_{\mu,k}(\xi_\mu)\geq\overline{V}^2_{>2\a}
\frac{4\a}{\sqrt{\hat{k}^2-4\a^2}(\sqrt{\hat{k}(\hat{k}-2\a)}+\sqrt{\hat{k}(\hat{k}+2\a)})}\geq \frac{c_2}{\sqrt{\al}}
\overline{V}^2_{>2\a}\,,
\ee
so the assertion follows combining the latter two with \eqref{combino}.
\end{proof}

\begin{rmk}
Note that if $V_{\rceil 2\al\rceil}\ne0$ there is a further $1/\sqrt{\upepsilon}$ in the lower bound 
in \eqref{stimeK2}.
\end{rmk}

\begin{lemma}\label{monte}
There is $c>0$ such that for all $\mu\in\Z$ and all $\x\in(\x_\mu -r , \x_\mu+r)$ if 
$\tilde{K}_2(\x)\ne0$ then one has
\begin{equation}\label{stimeK2}
\frac{c}{\sqrt\a}
\overline{V}^2
\le
\tilde{K}_2(\x) \le
c \sqrt{\frac{\a}{\upepsilon} } \|V\|^2_{L^2}.
\end{equation}
\end{lemma}

\begin{proof}
The lower bound follows exactly as the lower bound in Lemma \ref{lemma:bound-A}.
As for the upper bound we simply add to the upper bound for $K_2(\x)$  the quantity
\begin{equation}\label{pezzetto}
\begin{aligned}
\sum_{k=1}^{\lfloor 2\al\rfloor} |V_k|_2 A_{\mu,k}(\x) &\le c_1 \|V\|_{L^2}^2\left(
 \frac{r}{\sqrt{\al}} + \sum_{k=1}^{\lfloor 2\al\rfloor}\frac{4\a}{\sqrt{k^2-4\a^2}(\sqrt{k(k-2\a)}+\sqrt{k(k+2\a)})}
\right) \\
&\le c_2 \|V\|_{L^2}^2 \frac{8\al^2}{\sqrt{\upepsilon\lfloor 4\al\rfloor } ( \sqrt{\upepsilon\lfloor 2\al\rfloor }
+ 2\lfloor\al\rfloor)} \le c_3 \sqrt{\frac{\al}{\upepsilon}}\|V\|_{L^2}^2
\end{aligned}
\end{equation}
so the assertion follows.
\end{proof}

Now we are in position to prove Proposition \ref{prop:main}. 

\begin{proof}[Proof of Proposition \ref{prop:main}] 
The case $\mu=0$ is easier and can be studied separately.  
Indeed by a direct computation we see that 
\bea
D_0(\xi)G_{0,k}(\xi)&=&D^{-1}_{-k}(\x)\chi(\xi>1-r)+D^{-1}_{k}(\x)\chi(\xi<-1+r)\nn\\
&+&\ii\left(-D^{-1}_{-k}(\x)\chi(\xi<-1+r)+D^{-1}_{k}(\x)\chi(\xi>1-r)\right)\,.\nn
\eea

In particular $\Re( G_{0,k}(\xi))>0$ for all $k\in\Z$, so that
by (\ref{eq:MeG}) we have
\be\nn
|1-\mathcal M_0(\xi) Lj_0(\xi)|\geq 1-\Re\left(\mathcal M_0(\xi) Lj_0(\xi)\right)=
1+\g^2 \sum_{k\geq1}|V_k|^2\Re (G_{0,k}(\xi))>1\,. 
\ee

Now we study the case $\mu\geq1$. If $\xi \in(\xi_\mu-r,\xi_\mu]$ by Lemma \ref{lemma:G} we see that
\be
|1-\Re(\mathcal M_\mu(\xi)Lj_\mu(\xi))|=1+\g^2\sum_{k\geq1} |V_k|^2 \Re(G_{\mu,k})>1\,,
\ee
which entails
\be\label{eq:infM1}
\inf_{\xi \in(\xi_\mu,\xi_\mu+r)}|1-\mathcal M_\mu(\xi)Lj_\mu(\xi)|>1\,.
\ee

For all $\xi \in(\xi_\mu,\xi_\mu+r)$ we claim that
\be\label{eq:claim1/2}
|1-\mathcal M_\mu(\xi)Lj_\mu(\xi)|\geq\frac12\,,
\ee
for $\g$ small enough.
To prove it, we use again Lemma \ref{lemma:G}. 

We have
\bea
|1-\mathcal M_\mu(\xi)Lj_\mu(\xi)|^2&=&|1+\g^2\sum_{k\geq1}|V_k|^2\Re(G_{\mu,k}(\x))+
i(-\g \frac{V_0}{ D_\mu(\x)} + \g^2\sum_{k\geq1}|V_k|^2\Im(G_{\mu,k}(\x)))|^2\nn\\
&=&\left|1+\g^2\frac{\tilde{K}_2(\xi)}{D_\m(\x)}-\g^2\frac{K_1(\xi)}{D_\m(\x)}+\ii
(\g \frac{V_0}{D_\mu(\x)}+\g^2\frac{K_1(\xi)}{D_\m(\x)})\right|^2\nn\\
&=&\left(1-\g^2\frac{K_1(\xi)}{D_\m(\x)}+\g^2\frac{\tilde{K}_2(\xi)}{D_\m(\x)})^2+ (\g \frac{V_0}{D_\mu(\x)}+
\g^2\frac{K_1(\x)}{D_\m(\x)}\right)^2\,. \label{eq:marco}
\eea

Now if 
$$
\left|1-\frac{\g^2}{D_\m(\x)}(K_1(\xi)-\tilde{K}_2(\xi))
\right|\geq\frac12
$$ 
we have
\be\label{eq:claim1/2-1}
\mbox{r.h.s of (\ref{eq:marco})}\geq \frac14+\left(\g \frac{V_0}{D_\mu(\x)}+
\g^2\frac{K_1(\xi)}{D_\m(\x)}\right)^2\geq\frac14\,,
\ee
while if 
$$
\left|1-\frac{\g^2}{D_\m(\x)}(K_1(\xi)-\tilde{K}_2(\xi))\right|\leq\frac12
$$ 
then 
$$
\frac{\g^2}{D_\m(\x)}K_1(\xi)\geq \frac12+\frac{\g^2}{D_\m(\x)}\tilde{K}_2(\xi)
$$
 and moreover using Lemmata \ref{lemma:crux} and \ref{monte} we have also
 \begin{equation}\label{numeretto}
 \sqrt{\frac{\upepsilon}{\al}}\frac{1}{4\|V\|_{L^2}^2} \le
 \frac{\g^2}{D_\m(\x)}\le \frac{3\sqrt{\al}}{2( - \underline{V}^2_{\le 2\al} \sqrt{\upepsilon} + \ol{V}^2)}.
 \end{equation}
 
But then
\be\label{eq:claim1/2-2}
\mbox{r.h.s of (\ref{eq:marco})}\geq (\frac{\g}{D_\mu(\x)} {V_0}+
\frac{\g^2}{D_\mu(\x)} K_1(\x))^2
\geq(\frac12+\frac{\g}{D_\mu(\x)} {V_0}+\frac{\g^2}{D_\mu(\x)}\tilde{K}_2(\xi))^2 \ge \frac{1}{4}\,,
\ee
which is obvious if $V_0\ge0$ while if $V_0<0$ we need to impose
\begin{equation}\label{condi}
\g <  c \sqrt{\frac{\upepsilon}{\al}}\frac{|V_0|}{\|V\|_{L^2}^2}
\end{equation}
where $c$ is the constant appearing in Lemma \ref{monte}, in order to obtain
\[
\frac12+\frac{\g}{D_\mu(\x)} {V_0}+\frac{\g^2}{D_\mu(\x)}\tilde{K}_2(\xi) < - \frac{1}{2}
\]
Thus the assertion follows.
\end{proof}


\section{Proof of Proposition \ref{lemma:jR}}\label{stimazze}

Here we prove Proposition \ref{lemma:jR}.

\begin{proof}[Proof of Proposition \ref{lemma:jR}]
First of all we note that by (\ref{eq:jloc}) and (\ref{staqua}) we have
\be\label{eq:prima!}
\sup_{\xi\in(\xi_\mu-r,\xi_\mu+r)}|j_\mu^{\mathcal R}(\xi)|=\sup_{\xi\in(\xi_\mu-r,\xi_\mu+r)}
\left|\frac{j_\mu(\xi)}{1-\mathcal M_\mu(\x) j_\mu(\xi)}\right|=\left(\inf_{\xi\in(\xi_\mu-r,\xi_\mu+r)}
|D_\mu(\xi)-\mathcal M_\mu(\xi)|\right)^{-1}\,. 
\ee
Note that
$$
\mathcal M_\mu(\xi)=-\ii\g V_0-\g^2D_{\mu}(\xi)\sum_{k\geq1}|V_k|^2 G_{k,\mu}(\xi)\,.
$$

So thanks to Lemma \ref{lemma:G} and using the notation in \eqref{capponinuovi} we can write
\bea
\mathcal M_\mu(\xi)&=&-\ii\g V_0+\chi(\xi_\mu<\xi<\xi_\mu+r)\left(\g^2(K_1(\xi)-\tilde{K}_2(\xi))-
\ii\g^2K_1(\xi)\right)\nn\\
&-&\chi(\xi_\mu-r<\xi<\xi_\mu)\left(\ii\g^2(K_1(\xi)-\tilde{K}_2(\xi))+\g^2K_1(\xi)\right)\,.\label{eq:dec-M}
\eea
Therefore
\bea
|D_\mu(\xi)-\mathcal M_\mu(\xi)|&=&\chi(\xi_\mu<\xi<\xi_\mu+r)|D_\mu(\xi)-\g^2(K_1(\xi)-\tilde{K}_2(\xi))-
\ii(\g V_0+\g^2K_1(\xi))|\nn\\
&+&\chi(\xi_\mu-r<\xi<\xi_\mu)|D_\mu(\xi)+\g^2K_1(\xi)-\ii(\g V_0+\g^2(K_1(\xi)-\tilde{K}_2(\xi)))|\nn\,,
\eea
whence
\be\label{eq:combining}
\inf_{\xi\in(\xi_\mu-r,\xi_\mu+r)}|D_\mu(\xi)-\mathcal M_\mu(\xi)|\geq\min\left(\inf_{\xi\in(\xi_\mu-r,\xi_\mu+r)}d(\xi), \inf_{\xi\in(\xi_\mu-r,\xi_\mu+r)}s(\xi)\right)\,,
\ee
with
\bea
d(\xi)&:=&|D_\mu(\xi)-\g^2(K_1(\xi)-\tilde{K}_2(\xi))-\ii(\g V_0+\g^2K_1(\xi))|\label{eq:destra}\\
s(\xi)&:=&|D_\mu(\xi)+\g^2K_1(\xi)-\ii(\g V_0+\g^2(K_1(\xi)-\tilde{K}_2(\xi)))|\label{eq:sinistra}\,. 
\eea
To estimate $d(\xi)$ and $s(\xi)$ we treat separately the cases $V_0=0$ and $V_0\neq0$. 
Moreover for the first case we consider two sub-cases, namely either 
$\underline V_{\leq 2\a}\neq0$ or $\underline V_{\leq 2\a}=0$. 

\

\noindent
{\bf case I.1:} \underline{$V_0=0$, $\underline V_{\leq 2\a}\neq0$}. By  Lemma
 \ref{lemma:crux} there is a constant $c>0$ such that
\bea
\inf_{\xi\in(\xi_\mu-r,\xi_\mu+r)}d(\xi)&\geq&\g^2\inf_{\xi\in(\xi_\mu-r,\xi_\mu+r)}|K_1(\xi)|\geq c
\g^2\sqrt{\frac{\upepsilon}{\a}} \ \underline V^2_{\leq 2\a} \label{eq:bound1.1-d}\\
\inf_{\xi\in(\xi_\mu-r,\xi_\mu+r)}s(\xi)&\geq&\inf_{\xi\in(\xi_\mu-r,\xi_\mu+r)}|D_\mu(\xi)+\g^2K_1(\xi)|\nn\\
&\geq&\g^2\inf_{\xi\in(\xi_\mu-r,\xi_\mu+r)}|K_1(\xi)|\geq c
\g^2\sqrt{\frac{\upepsilon}{\a}} \ \underline V^2_{\leq 2\a}
\label{eq:bound1.1-s}\,.
\eea

\

\noindent
{\bf case I.2:} \underline{$V_0=0$, $\underline V_{\leq 2\a}=0$}. In this case $K_1(\x)=0$. 
On the other hand by  Lemma \ref{lemma:bound-A} 
there is a constant $c>0$ such that
\bea
\inf_{\xi\in(\xi_\mu-r,\xi_\mu+r)}d(\xi)&\geq&\inf_{\xi\in(\xi_\mu-r,\xi_\mu+r)}|D_\mu(\xi)+\g^2K_2(\xi)|\nn\\
&\geq&\g^2\inf_{\xi\in(\xi_\mu-r,\xi_\mu+r)}|K_2(\xi)|\geq  \g^2
\frac{c}{\sqrt\a} \overline{V}^2_{>2\a}
\label{eq:bound1.2-d}\\
\inf_{\xi\in(\xi_\mu-r,\xi_\mu+r)}s(\xi)&\geq&\g^2\inf_{\xi\in(\xi_\mu-r,\xi_\mu+r)}|K_2(\xi)|\geq \g^2
\frac{c}{\sqrt\a} \overline{V}^2_{>2\a}
\label{eq:bound1.2-s}\,.
\eea

Combining (\ref{eq:prima!}), (\ref{eq:combining}), (\ref{eq:bound1.1-d}), (\ref{eq:bound1.1-s}), (\ref{eq:bound1.2-d}), (\ref{eq:bound1.2-s}) gives the second line of (\ref{stastimazza}). 

\

\noindent
{\bf case II:} \underline{$V_0\neq0$}. We have
\bea
\inf_{\xi\in(\xi_\mu-r,\xi_\mu+r)}d(\xi)&\geq&\g\inf_{\xi\in(\xi_\mu-r,\xi_\mu+r)}|V_0+\g K_1(\xi)|\geq 
\g\frac{| V_0|}{2}\label{eq:bound2.1-d}\\
\inf_{\xi\in(\xi_\mu-r,\xi_\mu+r)}s(\xi)&\geq&\g\inf_{\xi\in(\xi_\mu-r,\xi_\mu+r)}|V_0+\g K_1(\xi)-\g K_2(\xi)|\nn\\
&\geq& \g\inf_{\xi\in(\xi_\mu-r,\xi_\mu+r)}|V_0-\g K_2(\xi)|\geq \frac{\g |V_0|}{2}\label{eq:bound2.1-s}\,.
\eea
The last inequality is always satisfied if $\overline V_{>2\a}=0$, while otherwise we need to require
\be\label{eq:gammaprovided1}
\g\leq \frac{1}{c}\sqrt{\frac{\upepsilon}{\al}}\frac{|V_0|}{\|V\|^2_{L^2}}
\ee
by Lemma \ref{lemma:bound-A}.

Combining (\ref{eq:prima!}), (\ref{eq:combining}), (\ref{eq:bound2.1-d}), (\ref{eq:bound2.1-s}), 
and (\ref{eq:gammaprovided1}) 
the result follows
\end{proof}

\medskip

\appendix{\section{Sketch of the proof of Proposition \ref{prop:asintotica}}\label{Appendix}

In this appendix we give a brief account on how to prove Proposition \ref{prop:asintotica}. 
We just outline the strategy and for most of the details we refer to our previous paper \cite{CG}. 
We consider only the case $V_0=0$, which is the most difficult case. 

For any function $F=F(t,x_1,x_2,\ldots)$ we write
$$
F=\OO{\sqrt t}\,\,\Longleftrightarrow\,\,\frac{C_1}{\sqrt{t}}\leq\sup_{x_1,x_2,\ldots}F\leq\frac{C_2}{\sqrt{t}} 
$$
for some $C_1,C_2>0$. 

Let us set for brevity
\be\label{eq:Delta}
\DD=\DD(\xi,t,t_0,\w):=\psi_{\infty}-\psi_{t_0}\,.
\ee
By (\ref{duat0}) and (\ref{duat}) we readily obtain the following equation
\be\label{eq:eqdelta}
(\uno+\ii h W_{t_0})\DD=-q^{[0]}\,,
\ee
where
\bea
q^{[0]}=q^{[0]}(\x,t,t_0,\w)&:=&\ii h (W_\infty-W_{t_0})\psi_{\infty}\label{eq:defq}\\
&=&\ii h \int_{-\infty}^{t_0}\dd \tau J_0(g(t-\tau))e^{\ii g\x(t-\tau)}V(\tau)\psi_\infty(\x,\tau)\,.\nn
\eea
}
Therefore we have to prove that for any $\w$ satisfying (\ref{eq:dioph-alpha}) there are $C_1,C_2>0$ such that
$$
\DD=\OO{\sqrt{t-t_0}} \,. 
$$
We write (here $\psi_{\infty, \mu}$ denotes the $\mu$-th Fourier coefficient of $\psi_\infty$)
\bea
q^{[0]}&=&\ii h\sum_{\mu,k\in\ZZZ}\psi_{\infty, \mu} V_k e^{\ii\w(\mu+k)t} \int_{t-t_0}^{\infty}\dd 
\tau J_0(g\tau) e^{\ii(g\xi-\w (\mu+k))\t}\nn\\
&=&\sum_{k\in\ZZZ}\ii h(\psi_{\infty}\ast V)_k e^{\ii\w kt} \int_{t-t_0}^{\infty}\dd \tau J_0(g\tau) 
e^{\ii(g\xi-\w k)\t}\,\nn\\
&=:&\sum_{k\in\ZZZ} e^{\ii \w k t}q^{[0]}_k(t,t_0,\xi,\w)\label{eq:q-dec}\,,
\eea
where the last line is understood as the definition of the coefficients $q^{[0]}_k(t,t_0,\xi,\w)=q_k^{[0]}$. 

The first property to establish is the decay of $q^{[0]}$. This is done first singularly on each
coefficient $q_k^{[0]}$ and then promoted to $q^{[0]}$ by analyticity. Using \cite[Lemma A.6]{CG} we compute
\be
\int_{t-t_0}^{\infty}\dd \tau J_0(g\tau) e^{\ii(g\xi-\w k)\t}= \frac{c(g)}{\sqrt{t-t_0}}\sum_{\s=\pm1}\s\frac{e^{\ii(g(\xi+\s)-\w k) (t-t_0)}}{\sqrt{|g(\xi+\s)-\w k|}}+\OO{(t-t_0)\sqrt{1-\t^2}}\,,
\ee
where $c(g)>0$ is a constant. The divergences appearing in the above formula get however cancelled. Indeed we observe by \eqref{duat}
$$
\d_{k,0}=\psi_{\infty,k}+ih(\psi_{\infty}\ast V)_k\int_{0}^{\infty}\dd \tau J_0(g\tau) e^{\ii(g\xi-\w k)\t}\,,
$$
so the divergences of 
$$
\int_{0}^{\infty}\dd \tau J_0(g\tau) e^{i(\xi-\w k)\t}
$$
must coincide with the zeros of $(\psi_{\infty}\ast V)_k$ and we can write
$$
\ii h(\psi_{\infty}\ast V)_k=-(\psi_{\infty,k}-\d_{k,0})\left( \int_0^\infty \dd \t J_0(g\t)e^{\ii (g\xi-\w k)\t} \right)^{-1}\,
$$
whence
\be\label{eq:q^[0]}
-q^{[0]}_{k}=(\psi_{\io,k}-\d_{k,0})\frac{\int_{t-t_0}^{\infty}\dd \tau J_0(g\tau) 
e^{\ii(g\xi-\w k) \t}}{\int_{0}^{\infty} \dd \t J_0(g\t)e^{\ii(g\xi-\w k)\t}}\,
\ee
is bounded with the desired decay (notice that the denominator can be computed as in \eqref{eq:jk-tau}). 
We obtain
\be\label{eq:combine}
q_k^{[0]}(t,t_0,\xi,\w)=(\psi_{\infty,k}-\d_{k,0}\psi_0)\left[\sum_{\s=\pm1}
\frac{e^{\ii(g\x-\w k+g\s)(t-t_0)}\tilde r(t,t_0,\xi,\w)}{\sqrt{t-t_0}}+\OO{t-t_0}\right]\,.
\ee

Next we invert the compact operator $\uno+\ii h W_{t_0}$ in a standard way. 
It suffices to prove that successive applications of $W_{t_0}$ preserve the decay of 
$q^{[0]}$. To this end we define
$$
q^{[1]}:=W_{t_0}q^{[0]}\,
$$
and represent (see \cite[Lemma 3.4]{CG})
\bea
q^{[1]}&=&\sum_{k\in\ZZZ}e^{\ii k\w t} q^{[0]}_k(\x,t,t_0,\w)\,,\label{eq:qexp2}\\
q_k^{[1]}(\x,t,t_0,\w)&:=&\int_0^{t-t_0} dt'J_0(gt')e^{\ii(g\xi-\w k)t'}( V\ast q^{[0]})_k\,.\nn
\eea
Combining (\ref{eq:combine}) with the last definition we get 
\be
q_k^{[1]}=\sum_{\mu\in\Z\,,\s=\pm1}V_{k-\mu}(\psi_{\infty,k}-\d_{n,0})
e^{\ii(g\xi-\w \m+\s)(t-t_0)}\int_0^{t-t_0} dt'\frac{J_0(gt')}{\sqrt{t-t_0-t'}}e^{\ii(\w(k-\m)-\s)t'}\,.
\ee
To evaluate the inner integral we use \cite[Lemma A.6]{CG} according to which
\be
\int_0^{t-t_0} dt'\frac{J_0(gt')}{\sqrt{t-t_0-t'}}e^{\i(\w(k-\m)-\s)t'}=
\begin{cases}
 \OOO{1}&\w(k-\m)\in\{-2g,0,2g\}\\
\OO{\sqrt {t-t_0}}&\mbox{otherwise}
\end{cases}\,. 
\ee
The first case never occurs, since $|\w(k-\m)|=2g$ is excluded by (\ref{eq:dioph-alpha}) and $\w(k-\m)=0$ 
since we are considering the case 
$V_0=0$. Therefore 
\be
q_k^{[1]}(t,t_0,\xi,\w)=\OO{\sqrt{t-t_0}}\,,
\ee
and again one can promote the decay to the entire expansion (\ref{eq:qexp2}) by analyticity.


\end{document}